\documentclass[11pt]{article}
\usepackage{makeidx}
\usepackage{euscript}
\usepackage{amsmath,amssymb,amsfonts,amsthm}
\usepackage{dsfont}
\usepackage{latexsym}
\usepackage{subfigure}
\usepackage{graphicx}
\usepackage{fancybox}
\usepackage{fullpage}
\usepackage[backref]{hyperref}
\usepackage{paralist}
\usepackage{color}
\usepackage{wrapfig}
\usepackage{tikz}
\usetikzlibrary{decorations.pathreplacing}
\usepackage{setspace}
\usepackage[ruled]{algorithm2e}
\usepackage[framemethod=tikz]{mdframed}
\usepackage{xspace}
\usepackage{pgfplots}
\usepackage{framed}
\pgfplotsset{compat=1.5}
\usepackage{todonotes}
\usepackage{verbatim}
\usepackage{fonttable}
\allowdisplaybreaks

\newcommand{\ksparse}{\mathcal F_s}
\newcommand{\ksparsebool}{\ksparse^{0/1}}

\newcommand{\calP}{\mathcal P}
\def\E{\mathop{\mathbb{E}}\displaylimits}
\newcommand{\cP}{\mathcal P}

\makeatletter
\@fleqnfalse
\@mathmargin\@centering
\makeatother

\renewenvironment{proof}{\noindent{\bf Proof : \ }}{\hfill$\Box$\par\medskip}

\newtheorem{theorem}{Theorem}[section]
\newtheorem{corollary}[theorem]{Corollary}
\newtheorem{lemma}[theorem]{Lemma}
\newtheorem{proposition}[theorem]{Proposition}
\newtheorem{definition}[theorem]{Definition}

\newtheorem{fact}[theorem]{Fact}

\newtheorem{problem}[theorem]{Problem}

\newenvironment{proofof}[1]{\begin{trivlist} \item {\bf Proof
#1:~~}}
  {\qed\end{trivlist}}
\renewenvironment{proofof}[1]{\par\medskip\noindent{\bf Proof of #1: \ }}{\hfill$\Box$\par\medskip}
\newcommand{\namedref}[2]{\hyperref[#2]{#1~\ref*{#2}}}

\def \FFST    {\mdef{\mathsf{FFST}}}

\def \sgn    {\mdef{\mathsf{sgn}}}
\def \M    {\mdef{\mathbf{M}}}
\def \X    {\mdef{\mathbf{X}}}
\def \Y    {\mdef{\mathbf{Y}}}
\def \tr    {\mdef{\mathsf{Tr}}}

\newcommand{\COMMENTED}[1]{{}}

\newcommand{\PPr}[1]{\ensuremath{\mathbf{Pr}\left[#1\right]}}
\renewcommand{\Pr}{\ensuremath{\mathbf{Pr}}}

\renewcommand{\O}[1]{\ensuremath{\mathcal{O}\left(#1\right)}}
\newcommand{\dist}{\mathrm{dist}\hspace{-1pt}}
\newcommand{\eps}{\epsilon}

\newcommand{\cE}{\mathcal{E}}

\newcommand{\F}{\mathbb{F}}

\newcommand{\mdef}[1]{{\ensuremath{#1}}\xspace}  
\newcommand{\myset}[1]{\mdef{\mathbb{#1}}}       

\DeclareMathOperator*{\argmax}{argmax}

\DeclareMathOperator*{\poly}{poly}
\DeclareMathOperator*{\median}{median}

\DeclareMathOperator*{\var}{Var}

\newcommand{\superscript}[1]{\ensuremath{^{\mbox{\tiny{\textit{#1}}}}}\xspace}
\def \th {\superscript{th}}     
\def \etal{{\it et~al.}}

\def \R        {\mdef{\mathbb{R}}}                   
\def \I        {\mdef{\myset{I}}}                    

\newcommand{\ignore}[1]{}

\newif\ifnotes\notestrue 
\ifnotes
\newcommand{\samson}[1]{\textcolor{purple}{{\bf (Samson:} {#1}{\bf ) }} \marginpar{\tiny\bf
             \begin{minipage}[t]{0.5in}
               \raggedright S:
            \end{minipage}}}            							
\else
\newcommand{\samson}[1]{}
\fi

\hypersetup{
    colorlinks   = true,
    citecolor    = blue,
	linkcolor		= red
}

\providecommand{\email}[1]{\href{mailto:#1}{\nolinkurl{#1}\xspace}}

\title{Fast Fourier Sparsity Testing}
\author{Grigory Yaroslavtsev\thanks{Indiana University, Bloomington \& The Alan Turing Institute, London, UK.
E-mail: \email{gyarosla@iu.edu}}
\and
Samson Zhou\thanks{Carnegie Mellon University \& Indiana University, Bloomington.
E-mail: \email{samsonzhou@gmail.com}}
}
\date{\today}

\begin{document}
\maketitle

\begin{abstract}
A function $f : \F_2^n \to \R$ is \emph{$s$-sparse} if it has at most $s$ non-zero Fourier coefficients. 
Motivated by applications to fast sparse Fourier transforms over $\F_2^n$, we study efficient algorithms for the problem of approximating the $\ell_2$-distance from a given function to the closest $s$-sparse function. 
While previous works (e.g., Gopalan \emph{et al.} SICOMP 2011) study the problem of distinguishing $s$-sparse functions from those that are far from $s$-sparse under Hamming distance, to the best of our knowledge no prior work has explicitly focused on the more general problem of distance estimation in the $\ell_2$ setting, which is particularly well-motivated for noisy Fourier spectra.
Given the focus on efficiency, our main result is an algorithm that solves this problem with query complexity $\O{s}$ for constant accuracy and error parameters, which is only quadratically worse than applicable lower bounds. 
\end{abstract}

\section{Introduction}
The \emph{Fourier representation} of the function $f : \F_2^{\,n} \to \R$ is the function $\hat{f} : \F_2^{\,n} \to \R$ defined by the forward Fourier transform $\hat{f}(\alpha) = \E_{x \in \F_2^{\,n}}[ f(x) \chi_\alpha(x)]$ and its inverse $f(x) = \sum_{\alpha \in \F_2^{\,n}} \hat{f}(\alpha) \chi_\alpha(x)$, where for each $\alpha \in \F_2^{\,n}$, the function $\chi_\alpha : \F_2^{\,n} \to \R$ is defined by $\chi_\alpha(x) = (-1)^{\sum_{i=1}^n \alpha_i x_i}$. 
The values $\hat{f}(\alpha)$ are the \emph{Fourier coefficients} of $f$. 
When $f$ has at most $s$ non-zero Fourier coefficients, we say that it is \emph{Fourier $s$-sparse}, or just \emph{$s$-sparse} for short.
The Fourier sparsity of functions plays an important role in many different areas of computer science, including error-correcting codes~\cite{GL89, AGS03}, learning theory~\cite{KM93, LMN93}, communication complexity~\cite{ZS09, BC99, MO09, TWXZ13}, property testing~\cite{GOSSW11,WY13}, and parity decision tree complexity~\cite{ ZS10, STV14}.

There has also been renewed interest in the Fourier sparsity of functions over various finite abelian groups with the recent development of
specialized Fourier transform algorithms for such functions~\cite{HIKP12a,HIKP12b}. 
These algorithms improve on the efficiency of the standard Fast Fourier Transform algorithms for functions with sparse Fourier transforms by taking advantage of this sparsity itself. 
Since many functions (and/or signals) in practical applications \emph{do} display Fourier sparsity, this line of research has yielded many exciting applications as well as theoretical contributions (see~\cite{SFFT13} for details). 
For example, much of the recent work on the sparse Fourier transform has focused on functions over fundamental domains, such as the line or the hypergrid. 
Meanwhile, a sparse Fourier transform for functions over $\F_2^{\,n}$ has been known for over twenty years as the Goldreich--Levin~\cite{GL89} and Kushilevitz--Mansour~\cite{KM93} algorithm. 
This algorithm can learn functions that are (close to) $s$-sparse, using time and query complexity $\poly(n, s)$.
Since many classes of functions over $\F_2^{\,n}$ are known to be close to being $s$-sparse for a certain value of $s$ (e.g., monotone functions, decision trees, $r$-DNF formulas, etc.), the sparse Fourier transform given by the GL/KM-algorithm is one of the cornerstones of computational learning theory. 

One of the main limitations of the sparse Fourier transform as a technique is the fact that its efficiency is conditional on the assumption that the data of interest can be sparsely represented in the Fourier domain.
Hence in order to reliably use sparse Fourier transform algorithms it is beneficial to have a way to \emph{test} if a function is $s$-sparse or, more generally, to estimate the distance of a function to the closest $s$-sparse function.
For such tasks \textit{property testing} algorithms often come into play as a preprocessing step (see, e.g., \cite{R08}) since they typically require a much smaller number of samples and other resources such as time and space.
An important consideration when using property testing is the fact that presence of two kinds of noise in the data must be tolerated: small fraction of errors/outliers~\cite{PRR06} (noise of small Hamming weight) as well as arbitrary noise with small $\ell_p$-norm~\cite{BRY14}.
Since the performance of sparse FFT algorithms is conditioned on the sparsity under $\ell_2^2$-distance, the subject of our study is to what extent can sparsity under $\ell_2^2$ distance be tested. 
The fundamental reason why $\ell_2^2$-distance plays a special role in the Fourier domain is its relation to the energy of the signal that is proportional to the sum of squares of the Fourier coefficients according to Parseval's theorem.

Formally, we define the \emph{$\ell_2^2$-distance} between $f$ and $g$ as $\dist_2^2(f,g) = \|f-g\|_2^2 = \frac{1}{2^{n}} \sum_{x \in \mathbb F_2^n} (f(x) - g(x))^2$ and the $\ell_2^2$-distance between $f$ and a class $\calP$ of functions as $\dist_2^2(f,\calP) = \min_{g \in \calP} \dist_2^2(f,g)$. 
The distance to Fourier $s$-sparsity is the latter distance when $\calP$ is the set of functions with Fourier sparsity at most $s$. 
We denote the class of all Fourier $s$-sparse functions as $\ksparse$. 
Hence our main goal is to estimate $\dist_2^2(f,\ksparse)$ up to an additive error $\pm \epsilon$.
We refer to it as $\ell_2^2$-distance estimation problem and to the closely related decision version as tolerant $\ell_2^2$-testing.

We also reserve the name of non-tolerant $\ell_2^2$-testing for an easier promise problem of distinguishing functions with Fourier sparsity at most $s$ from those that are $\epsilon$-far from having such sparsity, i.e., $\dist_2^2(f, \ksparse) \ge \epsilon$. 
Note that when working with noisy Fourier spectra, where most of the Fourier coefficients are non-zero, this decision can be trivial when $s \ll 2^n$, as an $\ell_2^2$-testing algorithm can just always reject. 
Hence the distance estimation problem described above can be substantially harder for such spectra. 
To simplify presentation, we call a class $\epsilon$-testable with $q$ queries if there exists an algorithm which makes $q$ queries and achieves the above guarantee with constant probability. 
We will also use Hamming distance while keeping the rest of the definitions the same in order to describe some of the previous work in the area of property testing. 
In this case the distance between $f$ and $g$ is defined as $\Pr_{x \sim \mathbb F_2^n}[f(x) \neq g(x)]$ and all the definitions above are changed accordingly.

\subsection{Previous work}
The most direct approach for $\ell_2^2$-distance estimation and $\ell_2^2$-testing is to use the testing-by-learning approach established by Goldreich, Goldwasser, and Ron~\cite{GGR98}. 
Using the Goldreich--Levin / Kushilevitz--Mansour algorithm~\cite{GL89,KM93}, we can learn an $s$-sparse function $h$ that will be essentially as close to $f$ as possible. 
We can then estimate the distance between $f$ and $h$ to get a good approximation of the distance from $f$ to Fourier $s$-sparsity.
This approach requires $\O{sn}$ queries in order to achieve constant error $\epsilon$ (see, e.g., the textbook exposition in~\cite{G01,O14}). 
An improvement to this approach would be to use hashing to reduce the dimension down to a subspace of size $\O{s^2}$ (thus introducing no collisions between the top $s$ coefficients) and then run GL/KM within the subspace. 
The complexity of this approach would be $\O{s \log s}$ queries for constant $\epsilon$, where the $\log s$ factor results from using $\O{s^2}$ buckets to avoid collisions among the top $s$ coefficients.
Other related previous work (e.g.~\cite{BBG18} who study testing sparsity over known and unknown bases, including the Fourier basis) also incurs extra factors in query complexity. 
\footnote{Also, since~\cite{BBG18} handles a much more general problem in order to handle arbitrary design matrices, the running time of their algorithms translates to polynomial in $2^n$ in our case, which can be prohibitively large for our application.}

The first specialized algorithm for the problem of testing Fourier sparsity under Hamming distance was developed by Gopalan et al.~\cite{GOSSW11}. They give a \emph{non-tolerant} tester for Fourier $s$-sparsity under Hamming distance with a number of queries to $f$ that is \emph{independent} of $n$ and \emph{polynomial} in $s$ and $1/\epsilon$. 
More precisely, the focus of~\cite{GOSSW11} was a slightly different problem where the class $\cP$ of interest is defined to contain only Boolean $s$-sparse functions. 
Below we will refer to this class as $\ksparsebool$.\footnote{While $\ksparsebool \subseteq \ksparse$ in general there is no known relationship between testing and distance estimation query complexities of classes and their subclasses.}
However, in fact~\cite{GOSSW11} show that with some loss in parameters, this problem can be reduced to estimating $\ell_2^2$-distance from $\ksparse$, the problem that we study in this paper.
Thus an implicit ingredient of the~\cite{GOSSW11} algorithm is a $\ell_2^2$-distance estimation algorithm for $\ksparse$ with query complexity $\O{\poly(s)}$ for any constant additive error. 

An active line of previous work focuses on \emph{tolerant} testing under Hamming distance. 
Wimmer and Yoshida~\cite{WY13} showed that the general approach of~\cite{GOSSW11} can be extended to yield tolerant testers for Fourier $s$-sparsity of Boolean functions.
Specifically, they give an algorithm that distinguishes between functions that are $\epsilon/3$-close to Fourier $s$-sparse from those that are $\epsilon$-far from Fourier $s$-sparse under Hamming distance, using poly($s$) queries.
This allows one to approximate the distance to Fourier $s$-sparsity up to some \emph{multiplicative} factor. The polynomial dependence on $s$ is fairly large and the result does not extend to additive error.
Algorithms for estimating the Hamming distance to Fourier $s$-sparsity up to an additive error can be also derived through a general framework of Hatami and Lovett~\cite{HL13}.
However, the instantiation of the~\cite{HL13} framework results in power tower dependency on $s$.

\subsection{Our Contributions}
We introduce two new algorithms for testing Fourier $s$-sparsity with respect to $\ell_2^2$-distance. Our first main result shows that one can approximate the distance to Fourier $s$-sparsity in $\ell_2^2$-distance with a number of non-adaptive queries that is in fact \emph{linear} in $s$. 
This result is proved in Section~\ref{sec:distance-approx}.

\begin{theorem}(Approximating $\ell_2^2$-distance to $s$-sparsity)\label{thm:l2-distance-approx}
For any $s \ge 1$ and $\epsilon > 0$, there is an algorithm that given non-adaptive query access to a function $f : \F_2^{\,n} \to \R$ with unit $\ell_2$-norm takes at most $\O{\frac{s}{\epsilon^4} \log \frac{1}{\epsilon} \log \frac{1}{\delta}}$ queries and approximates $\dist_2^2(f, \ksparse)$ up to an additive error $\pm\epsilon$ with probability $1 - \delta$ and running time $\tilde{\mathcal{O}}\left(\frac{s}{\epsilon^4}\log\frac{1}{\delta}\right)$ (see Section~\ref{sec:distance-approx}.)
\end{theorem}
Here, the $\tilde{\mathcal{O}}$ notation suppresses polylogarithmic factors in $s$ and $\frac{1}{\eps}$.

As mentioned before, the main challenge in testing Fourier $s$-sparsity with respect to $\ell_2^2$-distance instead of Hamming distance seems to be the accurate estimation of a large number of possibly small nonzero Fourier coefficients using a small number of queries.  
Whereas a function can only be $\eps$-far from Fourier $s$-sparsity with respect to Hamming distance by having a large number of nonzero Fourier coefficients, a function can be $\eps$-far from Fourier $s$-sparsity with respect to $\ell_2^2$-distance by either having too many large Fourier coefficients or a large number of small nonzero Fourier coefficients. 

Instead of estimating these small Fourier coefficients, we randomly partition the set of Fourier coefficients into a number of cosets by first picking a random subspace $H$ and measuring the energy (the sum of the squared Fourier coefficients) in each coset. 
If $H$ has sufficiently large codimension, then the top Fourier coefficients are partitioned into separate cosets, so the estimation of the energy in the top cosets is a good estimation of the energy of the top Fourier coefficients. 
To estimate the energy in each coset, we query the function at a number of random locations to obtain an empirical estimate within an additive factor of $\eps^2\|f\|_2^2$ with constant probability. 
We then bound the probability of two sources of errors: the hashing error, which originates from drawing a subspace in which large Fourier coefficients collide, and the estimation error, which results from inaccurate empirical estimations. 
Putting things together, we show that our estimator approximately captures the Fourier $s$-sparse function closest to $f$ in $\ell_2^2$-distance and hence gives a good approximation of the distance from $f$ to the closest Fourier $s$-sparse function.  

We also show a lower bound of $\Omega(\sqrt{s})$ for $\ell_2^2$-testing of $\ksparse$ for non-adaptive query algorithms. 
\begin{theorem}
For any $s \le 2^{n - 1}$, there exists a constant $c > 0$ such that any non-adaptive algorithm given query access to $f \colon \mathbb F_2^n \to \mathbb R$ such that $\|f\|_2^2 = 1\pm \epsilon$ that distinguishes whether $f$ is $s$-sparse or $f$ is $\frac 13$-far from $s$-sparse in $\ell_2^2$ with probability at least $2/3$ has to make at least $c\sqrt{s}$ queries to $f$ (see Section~\ref{sec:lb}).
\end{theorem}

Our lower bound results from designing two distributions $\mathcal D_{YES}$ and $\mathcal D_{NO}$, where the distribution $\mathcal D_{YES}$ is the set of Fourier $s$-sparse functions whose Fourier coefficients are scaled Gaussian random variables whereas the $\mathcal D_{NO}$ distribution is the set of functions with support on \emph{all} Fourier coefficients. 
The Fourier coefficients in the $\mathcal D_{NO}$ distribution are Gaussian random variables with a different scaling, such that the total variation distance between the $\mathcal D_{YES}$ and $\mathcal D_{NO}$ distributions restricted to a small query set is also small.

\cite{GOSSW11} gives an $\Omega(\sqrt{s})$ property testing lower bound for $\ksparsebool$. 
Their results can be extended to $\ksparse$, provided that $s \le 2^{cn}$ for a specific constant $c>0$, whereas our results covers the full range of values of $s$. 
Thus our results in Theorem~\ref{thm:l2-distance-approx} above are at most a quadratic factor away from optimal. 
We consider closing the quadratic gap in query complexity of $\ell_2^2$-distance estimation for $\ksparse$ as the main open problem posed by our work. 
\section{Preliminaries}
For a finite set $S$ we denote the uniform distribution over $S$ as $U(S)$.

\subsection{Fourier Analysis}
We consider functions from $\F_2^{\,n}$ to 
$\R$.
For any fixed $n \ge 1$, the space of these functions forms an inner product space
with the inner product
$\left<f, g\right> = \E_{x \in \F_2^{\,n}}[ f(x) g(x) ] = \frac1{2^n} \sum_{x \in \F_2^{\,n}} f(x)g(x)$.
The $\ell_2$-norm of $f  : \F_2^{\,n} \to \R$ is
$\| f \|_2 = \sqrt{ \left< f, f \right>} = \sqrt{\E_x[ f(x)^2 ]}$
and the $\ell_2$-distance between two functions $f, g : \F_2^{\,n} \to \R$ is 
the $\ell_2$-norm of the function $f - g$.
We write $\dist_2(f,g)=\lVert f-g \rVert_2$.
It is, in other words, $\|f - g \|_2 = \sqrt{\left< f-g, f-g \right>} 
= \frac1{\sqrt{|\F_2^{\,n}|}} \sqrt{\sum_{x \in \F_2^{\,n}} (f(x) - g(x))^2}$.

For $\alpha \in \F^n_2$, the \emph{character} 
$\chi_\alpha : \F^n_2 \to \{-1,1\}$ is the function defined by
$
\chi_\alpha(x) = (-1)^{\alpha \cdot x}.
$
The \emph{Fourier coefficient} of $f : \F_2^{\,n} \to \R$ corresponding to $\alpha$ is
$
\hat{f}(\alpha) = \E_x[ f(x) \chi_\alpha(x)].
$
The \emph{Fourier transform} of $f$ is the function $\hat{f} : \F_2^{\,n} \to \R$
that returns the value of each Fourier coefficient of $f$. 
The set of Fourier transforms of functions mapping $\F_2^{\,n} \to \R$ forms an inner
product space with inner product
$
\left< \hat{f}, \hat{g} \right> = \sum_{\alpha \in \F_2^{\,n}} \hat{f}(\alpha) \hat{g}(\alpha).
$
The corresponding $\ell_2$-norm is
$
\| \hat{f} \|_2 = \sqrt{\left< \hat{f}, \hat{f}\right>} = 
\sqrt{\sum_{\alpha \in \F_2^{\,n}} \hat{f}(\alpha)^2}.
$
Note that the inner product and $\ell_2$-norm are weighted differently for a function $f : \F_2^{\,n} \rightarrow \R$ and its Fourier transform $\hat{f} : \F_2^{\,n} \rightarrow \R$.  
We refer to the quantity $\hat{f}(\alpha)^2$ as the \emph{energy} of a Fourier coefficient $\hat{f}(\alpha)$.  

\begin{fact}[Parseval's identity]
	For any $f : \F_2^{\,n} \to \R$ it holds that
	$
	\|f \|_2 = \| \hat{f} \|_2 
	= \sqrt{ \sum_{\alpha \in \F_2^{\,n}} \hat{f}(\alpha)^2 }.
	$
\end{fact}

A function $f : \F_2^{\,n} \to\mathbb{R}$ is \emph{Fourier $s$-sparse} for some sparsity $s$ if the number of non-zero Fourier coefficients of $f$ is at most $s$. 
We let $\ksparse$ denote the set of all Fourier $s$-sparse functions. 

\subsection{Property Testing}
We study algorithms that make queries to a given function $f$. In this setting two different query access models are typicaly considered.
If all queries must be chosen in advance without access to the values of $f$, we call the corresponding algorithm \emph{non-adaptive} or equivalently, using \emph{non-adaptive queries}. 
Otherwise, the algorithm is \emph{adaptive}, and uses \emph{adaptive queries}, i.e. the queries made by the algorithm might depend on all previously queried values of $f$. 
In this paper, both our upper and lower bounds apply specifically to the non-adaptive query model.

We use the following standard definition of property testing under Hamming distance:

\begin{definition}[Property testing~\cite{GGR98}]
An algorithm $\mathcal A$ is a property tester with parameter $\epsilon>0$ for a class $\mathcal C$ of functions $f : \F_2^{\,n} \to \{-1,1\}$ if given query access to $f$ it distinguishes with probability at least $2/3$ whether $f \in \mathcal C$ or  $\min_{g \in \mathcal C} \Pr_{x \sim \mathbb F_2^n}[f(x) \neq g(x)] \ge \epsilon$. 
If neither of the two conditions hold then $\mathcal A$ can output an arbitrary answer.
\end{definition}

The notions of $\ell_2^2$-tester and distance approximator are 
defined below. In order to make $\epsilon$ be a scale-free parameter we assume that $\|f\|_2^2 = 1$ throughout this paper unless otherwise specified.
For example, for Boolean functions $f \colon \mathbb F_2^n \to \{-1,1\}$ this holds automatically and for real-valued functions this can be achieved by an appropriate scaling. 
The $\ell_2$-distance from a function $f : \F_2^{\,n} \to \R$ to
a class $\mathcal C$ of functions mapping $\F_2^{\,n}$ to $\R$ is 
$\dist_2(f, \mathcal C) = \min_{g \in \mathcal C} \|f - g\|_2$.

\begin{definition}[$\ell_2^2$-testing~\cite{BRY14}]
	An algorithm $\mathcal A$ is an \emph{$\ell_2^2$-tester} with parameter $\eps > 0$ for a 
	class $\mathcal C$ of functions $f : \F_2^{\,n} \to \R$ if given query access to $f$ with unit $\ell_2$-norm it distinguishes with probability at least $2/3$ whether 
	$f \in \mathcal C$
	or  $\dist_2^2(f, \mathcal C) \ge \epsilon$.
\end{definition}

In order to simplify presentation we say that a function $f$ is $\epsilon$-far from a class $\mathcal C$ in some distance (e.g. Hamming or $\ell_2^2$) if the closest function from $\mathcal C$ is at distance at least $\epsilon$ from $f$.

Generalizing the notion of $\ell_2^2$-testing we define a notion of $\ell_2^2$-distance approximation as follows:

\begin{definition}[$\ell_2^2$-distance approximator]\label{def:distance-approx}
An algorithm $\mathcal A$ is an \emph{$\ell_2^2$-distance approximator} with parameter $\eps > 0$	for a class $\mathcal C$ of functions $f : \F_2^{\,n} \to \R$ if given query access to $f$ with unit $\ell_2$-norm it outputs an estimate $\xi$ such that with probability at least $2/3$ it holds that $\left| \xi - \dist_2^2(f,\mathcal C) \right| \le \eps$.
\end{definition}

\subsection{Fourier Hashing}
\noindent
We use notation $H\le\F_2^{\,n}$ to denote a \emph{subspace} $H$ of $\F_2^{\,n}$. 
For $H \le \F_2^{\,n}$ we use notation $H^\perp$ for the \emph{orthogonal subspace} of $H$: $H^\perp:=\{z\in\F_2^{\,n}\,|\,\forall h\in H, z\cdot h=0\}$ where $\cdot$ denotes inner product for vectors. 
Given $a\in\F_2^{\,n}$, the \emph{coset} $a+H$ is defined by the set of points $a+H:=\{a+h|h\in H\}$. 
Note that a random subspace of dimension $d$ can be generated by selecting $d$ independent nonzero vectors of $\F_2^{\,n}$ uniformly at random. 
We say a subspace of $\F_2^{\,n}$ has \emph{codimension} $d$ if the subspace has dimension $n-d$. 

\begin{definition}\label{def:projected-function}
For a subspace $H \le \F_2^{\,n}$, an element $a \in H^\perp$, and a function $f : \F_2^{\,n} \to\mathbb{R}$, define the \emph{projected function} $f|_{a+H} : \F_2^{\,n} \to\mathbb{R}$ to be the function that satisfies 
$
f|_{a+H}(z) = \underset{x \in H^\perp}{\mathbb{E}}\big[ f(x+z) \chi_a(x) \big]
$
for each $z \in \F_2^{\,n}$. 
Given a subset $A \subseteq H^\perp$, we define $f|_{A+H} = \sum_{a \in A} f|_{a+H}$. 
\end{definition}

\noindent
From this definition, we observe that the values $f|_{a+H}(z)$ can all be computed simultaneously
\begin{proposition}\label{prop:coset-energy-estimation}
The set of queries $\{f(x+z)\}_{x \in H^\perp}$ can be used to compute $f|_{a+H}(z)$ for \emph{each} of the cosets $a+H$ of $H$ simultaneously.
\end{proposition}
We give more details about the number of queries required for computation of $f|_{a+H}$ in Lemma~\ref{lemma:accuracy:queries}.
We note that the projection of $f$ onto the cosets of a linear subspace $H$ yields a partition of the Fourier spectrum of $f$.
Moreover, the projection of $f$ to a coset $a+H$ is a function that zeroes out all Fourier coefficients not in $a + H$.

We now recall the following Poisson summation formula. 
For a reference, see Section 3.3 in \cite{O14}. 
We also give the proof of Proposition \ref{prop:projection-formulas} in Appendix~\ref{app:poisson}, for completeness.
\begin{proposition}[Poisson Summation Formula]
\label{prop:projection-formulas}
Fix any subspace $H \le \F_2^{\,n}$ and element $a \in \F_2^{\,n}$. 
Then for the projected function $f|_{a + H}$:
\begin{enumerate}
\item $f|_{a+H}(z) = \sum_{\beta \in a+H} \hat{f}(\beta) \chi_\beta(z)$ 
\item 
$
\widehat{f|}_{a+H}(\alpha) = \begin{cases}
\hat{f}(\alpha) & \mbox{if } \alpha \in a+H \\
0 & \mbox{otherwise.}
\end{cases}
$
\end{enumerate}
\end{proposition}
Proposition \ref{prop:projection-formulas} allows the following definition.
\begin{definition}
The total \textit{energy} of $f|_{a+H}$ is defined as $\sum_{\alpha \in a + H} \hat f(\alpha)^2 = \|\hat f|_{a + H}\|_2^2$.
\end{definition}

\begin{fact}
\label{fact:hash}
If $H \le \F_2^{\,n}$ is drawn uniformly at random from the set of subspaces of codimension $d$, then for any distinct $a, b \in \F_2^{\,n} \setminus \{0\}$, it holds that $\Pr[ b \in a+H] = 2^{-d}.$
\end{fact}

Fact~\ref{fact:hash} allows one to think of the projections $\{f|_{a+H}\}_{a \in H^\perp}$ as a hashing process applied to the Fourier coefficients of $f$. In fact, it is also known (for example, by Proposition 2.9 in \cite{GOSSW11}) that random projections correspond to a \emph{pairwise independent} hashing process.
\section{$\ell_2^2$-Distance Approximation and Sparsity Testing}\label{sec:distance-approx}
Recall that the property testing model, initiated by~\cite{GGR98}, requires an algorithm to accept objects that have some property $\mathcal{P}$ and reject objects that are at Hamming distance at least $\eps$ from having property $\mathcal{P}$ for some input parameter $\eps>0$ . 
In particular, in the property testing problem for $s$-sparsity, one would like to differentiate whether a given function $f \colon \mathbb\F_2^{\,n} \to \mathbb R$ with $||f||_2=1$ is in the class $\ksparse$ of Fourier $s$-sparse functions, or has distance at least $\eps$ from $\ksparse$. 

\begin{problem}[Property Testing for $s$-Sparsity]
Let $\ksparse$ be the class of $s$-sparse functions mapping from $\F_2^{\,n}$ to $\mathbb{R}$. 
Given query access to a function $f:\F_2^{\,n}\to\mathbb{R}$ with $||f||_2=1$ and parameter $\eps>0$, we call an algorithm $\mathcal{A}$ a property tester with query complexity $q$ if using at most $q$ queries, $\mathcal{A}$ accepts $f$ if $f\in\ksparse$ and rejects if $\underset{g\in\ksparse}{\min}||f-g||_2^2\ge\eps$.
\end{problem}

We now define the problem of energy estimation for the top $s$ Fourier coefficients, which also allows to solve the property testing problem. 
Note that this energy estimation problem for functions with unit $\ell_2$-norm is equivalent to the $\ell_2^2$-distance approximation problem in Definition~\ref{def:distance-approx} since both are defined in terms of additive error approximation.

\begin{problem}[Energy Estimation of top $s$ Fourier Coefficients]
Let $\ksparse$ be the class of $s$-sparse functions mapping from $\F_2^{\,n}$ to $\mathbb{R}$. 
Given non-adaptive query access to a function $f:\F_2^{\,n}\to\mathbb{R}$ with $||f||_2=1$ and parameters $s>0$ and $0<\eps\le 1$, we call an algorithm $\mathcal{A}$ an $\eps$-estimator of the energy of the top $s$ Fourier coefficients if using at most $q$ queries, $\mathcal{A}$ outputs $\xi$ such that
$
\left|\xi-\max_{|S|=s}\sum_{\alpha\in S}\hat{f}(\alpha)^2\right|\le\eps.
$
\end{problem}

The energy estimation problem can be used to solve the property testing problem above easily with roughly the same query complexity (see Fact~\ref{fact:approx-to-testing}).

Our Algorithm~\ref{alg:ee} estimates the energy of the top $s$ Fourier coefficients by first picking a random subspace $H$ of codimension $d= \log\frac{2s}{\eps^4}$ uniformly at random. 
The intuition is that by picking the codimension to be large enough, the top $s$ Fourier coefficients are partitioned into cosets with only a few collisions, so the estimation of the energy in the top $s$ cosets is a good estimation of the energy of the top $s$ Fourier coefficients. 
To estimate the energy in the top $s$ cosets, Algorithm~\ref{alg:ee} samples $\gamma=\O{\frac{s}{\eps^4}\|f\|_2^2}$ pairs $(x,x+z)$ to obtain an empirical estimate of the energy in each coset within an additive factor of $\eps^2\|f\|_2^2$ with constant probability. 
This yields the proof of Theorem~\ref{thm:l2-distance-approx}.  
Similarly, Algorithm~\ref{alg:ffst} gives a property tester for $s$-sparsity.
The success probability for each of these algorithms can be increased to $1-\delta$ for any $\delta>0$ by taking the median of $\O{\log\frac{1}{\delta}}$ parallel repetitions.

{\centering
\begin{minipage}{1\linewidth}
\begin{algorithm}[H]
\caption{\textsc{Energy Estimation}($\eps$, $s$)}\label{alg:ee}
Draw $H \le \F_2^{\,n}$ of codimension $d = \log\frac{2s}{\eps^4}$ uniformly at random\;
\For {$j = 1$ to $\ell = \O{\log\frac{1}{\eps}}$}
{
$\mathcal I_j \leftarrow$ set of pairs $(x, x+z)$ of size $\gamma = \O{\frac{s}{\eps^4}\|f\|_2^2}$, where $x \sim U(\F_2^{\,n}),z \sim U(H^\perp)$\;
\For{each $a \in H^\perp$}
{
$y_{a+H}^{(j)} \leftarrow 0$\;
\For{each $(x,x+z)\in \mathcal I_j$}
{
$y_{a+H}^{(j)}\leftarrow y_{a+H}^{(j)}+\frac{1}{|\mathcal I_j|}\chi_a(z)f(x)f(x+z)$
}
}
}
Return: $\xi:=\max_{S \subseteq H^\perp : |S| = s} \sum_{a \in S}\median\left( y_{a + H}^{(1)}, y_{a + H}^{(2)}, \dots, y_{a + H}^{(\ell)}\right)$.
\end{algorithm}
\end{minipage}
\par}

{\centering
\begin{minipage}{1\linewidth}
\begin{algorithm}[H]
\caption{\textsc{Fast Fourier Sparsity Test (\FFST)}($\eps,s$)}\label{alg:ffst}
Let $f$ be some function with known $||f||_2$. 
\vskip 0.05in\noindent 
Let $\xi$ be the output of Algorithm~\ref{alg:ee} on input $\frac{\eps}{2}$ and sparsity $s$.
\vskip 0.05in\noindent 
If $\xi\le \left(1 - \frac{\eps}{2}\right)||f||_2^2$, reject.
\vskip 0.05in\noindent
Otherwise, accept.
\end{algorithm}
\end{minipage}
\par}

Our analysis deals with two possible sources of error in the energy estimation. 
In Section~\ref{sec:hashing:error}, we consider the error caused by collisions in the hashing scheme and in Section~\ref{sec:estimation:error}, we consider the error caused by sampling variance in the energy estimates. 
Note that we perform worst-case analysis (over all possible sets of size $s$) for the hashing error as in the last step of the algorithm we adaptively select the largest subset.

\subsection{Hashing Error}
\label{sec:hashing:error}
We first analyze the error introduced into our estimator by hashing the Fourier coefficients across multiple cosets (assuming all estimates of energies in the cosets are exact).
Thus the first technical component of the analysis of the sparsity distance approximator shows that for a random choice of subspace $H$ of codimension $\log\frac{2s}{\eps^4}$, the union of the top $s$ cosets of $H$ has total energy that is close to the sum of the Fourier mass of the $s$ coefficients largest in magnitude.

Let $\cE_1 \ge  \dots \ge \cE_{2^n}$ be the true values of the energies of the $2^n$ Fourier coefficients corresponding to the function $f : \F_2^{\,n} \to\mathbb{R}$.
Let $h$ be some pairwise independent hash function with domain $[2^n]$ and range $[2^d]$, which can be viewed as partitioning the $2^n$ Fourier coefficients across the $2^d$ cosets, which we refer to as buckets. 
We denote the overall energy in the $i$-th bucket as $y_i$, where we assume that the hash function is clear from the context. Let the buckets be indexed in the non-increasing order by energy, so that $y_1 \ge y_2 \ge \dots \ge y_{2^d}$. 
Furthermore, let $y_i^*$ denote the energy of the largest coefficient hashing into the $i$-th bucket.
Formally, if index $i$ corresponds to coset $a+H$, then we let
\begin{align*}
&y_i = \sum_{\beta \in a + H} \hat f(\beta)^2, 
&& y_i^* = \max_{\beta \in a + H} \hat f(\beta)^2.
\end{align*}

\begin{definition}[Hashing Error]
We define the \emph{hashing error} of $h$ as
$err^s_h(\cE_1, \dots, \cE_{2^n}) = \sum_{i = 1}^s y_i - \cE_i$ to be the difference between the overall energy in the top $s$ buckets and the energy of the top $s$ coefficients.
\end{definition}
Note that the hashing error is always non-negative as there are at most $s$ buckets containing the top $s$ Fourier coefficients.
The contribution to the energy of the $i$\th bucket from the largest Fourier coefficient hashing into this bucket is denoted as $y^*_i$.
We have: 
\begin{align}
\label{eqn:error:bound}
err^s_h(\cE_1, \dots, \cE_{2^n})=\sum_{i =1}^s y_i - \cE_i = \sum_{i =1}^s y_i -y^*_i + \sum_{i = 1}^s y^*_i - \cE_i \le \sum_{i =1}^s y_i -y^*_i,
\end{align}
where we used the fact that $\sum_{i = 1}^s y_i^* \le \sum_{i = 1}^s \cE_i$. 

We can bound the hashing error across any set of $s$ buckets, rather than just the hashing error across the buckets containing the top $s$ Fourier coefficients. 
\begin{lemma}[Expected Hashing Error Bound]\label{lem:exp-hashing-error}
Let $H \le \F_2^{\,n}$ be a subspace of codimension $d$ drawn uniformly at random. 
Let $z_i = y_i - y^*_i$ be the ``collision error'' in the $i$\th bucket. 
Then
$$
\underset{H}{\mathbb E}\left[\sum_{i =1}^s z_i\right] \le \sqrt{\frac{2s}{2^d}}||f||_2^2.
$$
\end{lemma}
\begin{proof}
By the Cauchy-Schwarz inequality, $\sum_{i =1}^s z_i \le \sqrt{s}\sqrt{\sum_{i =1}^s z_i^2}$.
Let $\delta_{jk}$ be the indicator variable for the event that Fourier coefficients $\cE_j$ and $\cE_k$ collide and let $D_j$ be the indicator variable for the event that $\cE_j$ is not the largest coefficient in its hash bucket.
Then we have: 
$$\sum_{i =1}^s z_i^2 \le \sum_{i = 1}^{2^d} z_i^2 = \sum_{i= 1}^{2^d} (y_i - y^*_i)^2 = \sum_{j,k \in [2^n]} \cE_j \cE_k \delta_{jk} D_j D_k \le \sum_{j,k \in [2^n]} \cE_j \cE_k \delta_{jk} D_j,$$ 
where the first inequality holds since the $s$ buckets is a subset of the $2^d$ buckets, and the second inequality holds because either $D_k=0$ or $D_k=1$. 

Taking expectation over $H$ we have:
\begin{align*}
\underset{H}{\mathbb E}\left[\sum_{j,k \in [2^n]} \cE_j \cE_k \delta_{jk} D_j\right] = \underset{H}{\mathbb E}\left[\sum_{j \in [2^n]} \cE^2_j D_j\right] + \underset{H}{\mathbb E}\left[\sum_{j \neq k \in [2^n]} \cE_j \cE_k \delta_{jk}D_j\right] \le \underset{H}{\mathbb E} \left[\sum_{j \in [2^n]} \cE^2_j D_j\right] + \frac{\left(\sum_{j = 1}^{2^n} \cE_j\right)^2}{2^d},
\end{align*}
where we used Fact~\ref{fact:hash} and pairwise independence, so that $\underset{H}{\mathbb E}[\delta_{jk}] =\frac{1}{2^d}$. 
Note that by Fact~\ref{fact:hash}, pairwise independence and a union bound, $\Pr[D_j=1] \le \frac{j - 1}{2^d}$ and hence for the first term we have:
\begin{align*}
\underset{H}{\mathbb E}\left[\sum_{j \in [2^n]} \cE^2_j D_j\right] \le \sum_{j = 1}^{2^n} \frac{j - 1}{2^d} \cE^2_j \le \frac{1}{2^d} \sum_{j = 1}^{2^n} \sum_{k = 1}^{j - 1} \cE_j \cE_k \le \frac{\left(\sum_{j = 1}^{2^n} \cE_j\right)^2}{2^d}
\end{align*}
Putting things together, we have 
$$
\underset{H}{\mathbb E} \left[\sum_{i =1}^s z_i\right] \le \sqrt{s}\cdot\underset{H}{\mathbb E} \left[\sqrt{\sum_{i =1}^s z_i^2} \right] \le \sqrt{s}\cdot \sqrt{\underset{H}{\mathbb E} \left[\sum_{i =1}^s z_i^2 \right]}\le \sqrt{\frac{2s}{2^d}} \sum_{j = 1}^{2^n} \cE_j=\sqrt{\frac{2s}{2^d}}||f||_2^2,
$$
where we recall that the first inequality is by Cauchy-Schwarz, the second is by Jensen and the third is from the bound on $\underset{H}{\mathbb E} \left[\sum_{i =1}^s z_i^2 \right]$ derived above.
\end{proof}

\noindent
Now we give an upper bound on the variance of the difference between the energies of the top $s$ buckets and their respective largest Fourier coefficients:
\begin{lemma}[Variance of the Hashing Error]\label{lem:var-hashing-error}
Let $H \le \F_2^{\,n}$ be a subspace of codimension $d$ drawn uniformly at random. 
Let $z_i = y_i - y^*_i$ be the ``collision error'' in the $i$\th bucket. 
Then
$$
\underset{H}{\var}\left[\sum_{i = 1}^s z_i\right] \le \frac{2||f||_2^4}{2^d}.
$$
\end{lemma}
\begin{proof}
By pairwise independence we have:
\begin{align*}
\underset{H}{\var}\left[\sum_{i = 1}^s z_i\right] &\le\underset{H}{\var}\left[\sum_{i = 1}^{2^d} z_i\right] = \sum_{i = 1}^{2^d} \underset{H}{\var}[z_i] \le \sum_{i = 1}^{2^d} \underset{H}{\mathbb{E}}\left[z^2_i\right]\le \frac{2(\sum_{i = 1}^n \cE_i)^2}{2^d}=\frac{2||f||_2^4}{2^d},
\end{align*}
where the last inequality follows using the same argument as in the proof of Lemma~\ref{lem:exp-hashing-error}.
\end{proof}

\noindent
We now give a bound on the hashing error. 
\begin{corollary}\label{cor:hashing-error}
If $2^d = \frac{2s}{\eps^4}$  and $0 < \eps \le 1/2$, then 
$$\underset{H}{\mathbf{Pr}}\Big[err^s_h(\cE_1, \dots, \cE_{2^n})\le 5 \eps^2||f||_2^2\Big]\ge\frac{15}{16}.
$$
\end{corollary}
\begin{proof}
Recall that $err^s_h(\cE_1, \dots, \cE_{2^n})=\sum_{i =1}^s y_i - \cE_i$. 
Let $z_i = y_i - y^*_i$ be the collision error in the $i$\th bucket and let $Z = \sum_{i = 1}^s z_i$ and $\|\cE\|_1 = \sum_{i = 1}^n \cE_i$.
From Lemma~\ref{lem:exp-hashing-error}, Lemma~\ref{lem:var-hashing-error}, and Chebyshev's inequality, we have that for any $\alpha > 0$:
$$\Pr\left[Z \ge \sqrt{\frac{2s}{2^d}}||f||_2^2 + \alpha \sqrt{\frac{2}{2^d}}||f||^2_2\right] \le \frac{1}{\alpha^2}.$$
For $2^d = \frac{2s}{\eps^4}$ we have $\Pr[Z \ge (1 +\frac{\alpha}{\sqrt{s}} )\eps^2||f||_2^2] \le 1/\alpha^2$. 
Recall from (\ref{eqn:error:bound}):
$$
\sum_{i =1}^s y_i - \cE_i = \sum_{i =1}^s y_i -y^*_i + \sum_{i = 1}^s y^*_i - \cE_i \le \sum_{i =1}^s y_i -y^*_i,
$$
since $\sum_{i = 1}^s y_i^* \le \sum_{i = 1}^s \cE_i$. 
Taking $\alpha=4$ and noting that $s\ge 1$, it follows that
$$
\sum_{i =1}^s y_i - \cE_i\le 5 \eps^2||f||_2^2
$$
with probability at least $15/16$.
\end{proof}

\subsection{Estimation Error}
\label{sec:estimation:error}
We now analyze the error introduced to our estimator through sampling used to approximate the true bucket energies. 
Our intuition stems from the following standard fact to estimate the total energy via sampling. 

\begin{fact}
[Fact 2.5 in \cite{GOSSW11}]
\label{fact:energy:estimate}
$\sum_{\alpha\in a+H}\hat{f}(\alpha)^2=\underset{x\in\F_2^{\,n},z\in H^\perp}{\mathbb{E}}\left[\chi_a(z)f(x)f(x+z)\right]$.
\end{fact}

Using Fact~\ref{fact:energy:estimate}, the energy $\sum_{\alpha\in a+H}\hat{f}(\alpha)^2$ in each bucket $a+H$ can be approximated by repeatedly querying $f$ using the following Lemma~\ref{lemma:accuracy:queries}, whose proof is similar to Proposition 2.6 in \cite{GOSSW11}. 
We include the full proofs to formalize the dependency on $||f||_2$. 

In the language of Lemma~\ref{lemma:accuracy:queries}, suppose $y_i$ is the energy of bucket $a+H$ and $\mathcal{I}_j$ is a set of pairs $(x,x+z)$ of size $\gamma$, as in Algorithm~\ref{alg:ee}. 
Then the estimate $y_{i,j}$ corresponding to a sample $\mathcal{I}_j$ is: 
$$y_{i,j} = \frac{1}{|\mathcal{I}_j|}\sum_{(x,x+z)\in \mathcal{I}_j}\chi_a(x)f(z)(x+z).$$

We now bound the expected squared distance between $y_{i,j}$ and $y_i$ by the inverse of the sample size. 
\begin{lemma}
\label{lemma:accuracy:queries}
Given a subspace $H\le\F_2^{\,n}$, let $y_1\ge y_2\ge\ldots\ge y_{2^d}$ be the true energies in each of the buckets and $y_{i,j}$ be the estimate of $y_i$ given sample $\mathcal{I}_j$ of size $\gamma$. 
Then using $\gamma=\O{\frac{s}{\eps^4}||f||_2^2}$ queries to $f$,
$$
\underset{\mathcal{I}_j}{\mathbb{E}}\Big[|y_{i,j} - y_i|^2\Big] \le\frac{\eps^4}{s}||f||^2_2.
$$
\end{lemma}
\begin{proof}
Given a subspace $H\le\F_2^{\,n}$, let $x,y\in\F_2^{\,n}$ so that $|f(x)f(y)|\le\frac{f^2(x)+f^2(y)}{2}\le\frac{1}{2}||f||^2_2$. 
Thus, an empirical estimation of $\underset{x\in\F_2^{\,n},z\in H^\perp}{\mathbb{E}}\left[\chi_a(z)f(x)f(x+z)\right]$, with $\O{\frac{1}{\eps^2}\log\frac{1}{\delta}}$ queries to $f$, is within an additive factor of $\eps||f||^2_2$ with probability at least $1-\delta$ by standard Chernoff bounds. 

Let $C$ be a constant such that $\frac{C}{\eps^2}\log\frac{1}{\delta}$ samples suffice to estimate $y_{i,j}-y_i$ within an additive $\eps||f||_2^2$ with probability at least $1-\delta$. 
Equivalently for any $\theta>0$, the probability that $|y_{i,j}-y_i|\ge\theta$ using $\gamma$ samples is at most $e^{-\frac{\gamma\theta^2}{C||f||_2^4}}$. 
Then we have:
\begin{align*}
\mathbb{E}\left[|y_{i,j}-y_i|^2\right]&=\int_0^\infty\PPr{|y_{i,j}-y_i|^2\ge t}\,dt\\
&=\int_0^\infty\PPr{|y_{i,j}-y_i|\ge\sqrt{t}}\,dt\\
&\le\int_0^\infty e^{-\frac{\gamma t}{C||f||_2^4}}\,dt=\frac{C||f||_2^4}{\gamma}.
\end{align*}
Hence, for $\gamma=\O{\frac{s}{\eps^4}||f||_2^2}$, we have $\mathbb{E}\left[|y_{i,j} - y_i|^2\right] \le\frac{\eps^4}{s}||f||^2_2$, as desired.
\end{proof}

Note that the estimate $y_{i,j}$ is exactly the estimate $y^{(j)}_{a+H}$ in Algorithm~\ref{alg:ee}, where bucket $a+H$ is the bucket with the $i\th$ largest Fourier coefficient. 
We use two different notations to refer to the same quantity since it is more convenient to use the notation $y_{i,j}$ to index estimates by magnitude of Fourier coefficient, whereas the notation $y^{(j)}_{a+H}$ is more convenient to index by coset. 
Moreover, observe that we can obtain estimates $y_{i,j}$ of the energies $y_i$ simultaneously, by Proposition~\ref{prop:coset-energy-estimation}.

As before, let $y^*_i$ denote the contribution to the energy of the $i$\th bucket from the largest Fourier coefficient hashing into this bucket.
\begin{lemma}
\label{lemma:energies:samples}
Let $\epsilon>0$ and $H$ be a random subspace of codimension $d=\log\frac{2s}{\eps^4}$ and let $y_1\ge y_2\ge\ldots\ge y_{2^d}$ be the true energies in each of the buckets. 
Let $\ell=\O{\log\frac{1}{\eps}}$ be the number of random samples. 
Then for any $\eta > 0$, 
$$
\Pr[|y^*_i - y_i|^2\ge\eta] \le \left(\frac{2e\eps^4||f||^2_2}{s\eta}\right)^{\ell/2},
$$
where the probability is taken over all samples of size $\ell$.
\end{lemma}
\begin{proof}
By applying Markov's inequality to Lemma~\ref{lemma:accuracy:queries}, it follows that for each pair of $i$ and $j$,
$$
\Pr\left[|y_{i,j} - y_i|^2\ge\eta\right]\le\frac{\eps^4||f||^2_2}{s\eta}.
$$
Then the probability that at least half of the $\ell$ samples returns such estimates is
\begin{align*}
\Pr\left[|\{j : |y_{i,j} - y_i|^2 \ge\eta\}| > \frac{\ell}{2}\right] 
\le \binom{\ell}{\ell/2} \left(\frac{\eps^4||f||^2_2}{s\eta}\right)^{\ell/2}\le \left(\frac{2e\eps^4||f||^2_2}{s\eta}\right)^{\ell/2}, 
\end{align*}
where the second inequality follows from the well-known bound on the binomial coefficient $\binom{n}{k}\le\left(\frac{n\cdot e}{k}\right)^k$ for all $1\le k\le n$. 
The claim then follows from the fact that $y^*_i$ is the median of $y_{i,j}$ across all $j$. 
\end{proof}

\begin{lemma}
\label{lem:median:estimate}
Let $H$ be a random subspace of codimension $d=\log\frac{2s}{\eps^4}$. 
Then the expected value of the estimation error satisfies
$$
\underset{H}{\mathbb{E}}\left[\sum_{i=1}^s|y^*_i - y_i|^2\right]\le\eps^2\cdot||f||^2_2.
$$
\end{lemma}
\begin{proof}
Let $\beta=\frac{2e\eps^4||f||^2_2}{s\eps^{4/\ell}}$. 
Then:
\begin{align*}
\underset{H}{\mathbb{E}}\left[\sum_{i=1}^s|y^*_i - y_i|^2\right] &= \underset{H}{\mathbb{E}}\left[\int_0^\infty \min(s, |\{a \colon |y^*_a - y_a|^2\ge \eta |\})d\eta \right] \\
& \le \int_0^\infty \min(s, \mathbb E\left[|\{i \colon |y^*_i - y_i|^2\ge \eta |\}\right]) d\eta \\
& \le \int_0^\infty \min\left(s, 2^d\left(\frac{2e\eps^4||f||^2_2}{s\eta}\right)^{\ell/2}\right) d\eta \\
& \le \int_0^\beta s\,d\eta+\int_\beta^\infty 2^d\left(\frac{2e\eps^4||f||^2_2}{s\eta}\right)^{\ell/2} d\eta,
\end{align*}
where the second inequality follows from Lemma~\ref{lemma:energies:samples}.
Thus, 
\begin{align*}
\underset{H}{\mathbb{E}}\left[\sum_{i=1}^s|y^*_i - y_i|^2\right] &\le\frac{2e\eps^4||f||^2_2}{\eps^{4/\ell}}+2^d\left(\frac{2e\eps^4||f||^2_2}{s}\right)^{\ell/2}\frac{2}{\ell-2}\left(\frac{1}{\beta}\right)^{\ell/2-1}\\
&=\frac{2e\eps^4||f||^2_2}{\eps^{4/\ell}}+2^d\left(\frac{2e\eps^4||f||^2_2}{s}\right)\frac{\eps^2}{\eps^{4/\ell}}\frac{2}{\ell-2}.
\end{align*}
Hence for $\ell=\Theta\left(\log\frac{1}{\eps}\right)$, we have 
$
\underset{H}{\mathbb{E}}\left[\sum_{i=1}^s|y^*_i - y_i|^2\right]\le\eps^2\cdot||f||^2_2.
$
\end{proof}

\subsection{Proof of Theorem~\ref{thm:l2-distance-approx}} 
\label{sec:tt}
Recall that our algorithm returns an estimate $\xi$ of the sum of the $s$ buckets with the largest energy. 
Since the estimation error is small by Lemma~\ref{lem:median:estimate}, $\xi$ is a good estimate of the actual sum of the $s$ buckets with the largest energy. 
Because the hashing error is small by Corollary~\ref{cor:hashing-error}, $\xi$ is also a good approximation of the energy of the $s$ Fourier coefficients $\beta_1,\ldots,\beta_s\in\F_2^{\,n}$ with the largest energy. 
We define the function $f^*$ so that the Fourier transforms of $f^*$ and $f$ have the same values at the Fourier coefficients $\{\beta_i\}$. 
However, the Fourier transform of $f^*$ has value zero at the remaining coefficients outside of $\{\beta_i\}$. 
Thus by Parseval's identity, $f^*$ is the $s$-sparse function closest to $f$. 
Hence, $\xi$ is a good estimate of $||f^*||_2^2$. 

For each random sample $\mathcal{I}_j$ of size $\gamma=\O{\frac{s\|f\|_2^2}{\eps^4}}$, let $y_{a+H}^{(j)}$ be the corresponding estimate of $(\hat{f}|_{a+H})^2$. 
Let $S^*=\argmax_{|S|=s}\sum_{a\in S}\median\{y_{a+H}^{(1)},y_{a+H}^{(2)},\ldots,y_{a+H}^{(\ell)}\}$, where $\ell=\O{\log\frac{1}{\eps}}$ is the number of repetitions. 
Let $\beta^*_{f|a+H}=\argmax_{\alpha\in a+H}\hat{f}(\alpha)^2$ and define the function $h:\F_2^{\,n}\to\mathbb{R}$ by setting
$$
\hat{h}(\beta^*_{f|a+H})=\sgn(\hat{f}(\beta^*_{f|a+H}))\cdot\median\left\{\sqrt{y_{a+H}^{(i)}}\right\}
$$
for each $a\in S^*$ to be the only non-zero Fourier coefficients of $h$. 
Let $\beta_1,\beta_2,\ldots,\beta_s$ be defined so that 
$$
\hat{f}(\beta_1),\hat{f}(\beta_2),\ldots,\hat{f}(\beta_s)
$$
are the largest $s$ Fourier coefficients of $f$. 
Define the function $f^*:\F_2^{\,n}\to\mathbb{R}$ by setting
$$
\hat{f^*}(\beta_i)=\hat{f}(\beta_i)
$$
for each $1\le i\le s$ to be the only non-zero Fourier coefficients of $f^*$. 
\begin{lemma}
\label{lem:output:acc:first}
Let $\xi$ be the output of Algorithm~\ref{alg:ee} and $f^*$ and $h$ be defined as above. 
Then
$$
\Big|\xi-||f^*||_2^2\Big|\le2||f^*-h||_2||f||_2.
$$
\end{lemma}
\begin{proof}
Observe that Algorithm~\ref{alg:ee} outputs
$$
\xi=\sum_{a\in S^*}\median\left\{y_{a+H}^{(1)},y_{a+H}^{(2)},\ldots,y_{a+H}^{(i)}\right\}=\sum_{a\in S^*}\hat{h}(\beta^*_{f|a+H})^2=||h||_2^2.
$$
Therefore,
\begin{align*}
\Big|\xi-||f^*||^2_2\Big|=\left|||h||_2^2-||f^*||^2_2\right|=\left(\left|||h||_2-||f^*||_2\right|\right)\left(\left|||h||_2+||f^*||_2\right|\right).
\end{align*}
By triangle inequality, $\left|||h||_2-||f^*||_2\right|\le||f^*-h||_2$ and $\left|||h||_2+||f^*||_2\right|\le||h||_2+||f^*||_2$.
Thus, 
$$
\Big|\xi-||f^*||^2_2\Big|\le||f^*-h||_2(||h||_2+||f^*||_2).
$$
Since $||h||_2+||f^*||_2\le2||f||_2$, then it remains to bound $||f^*-h||_2$.
\end{proof}

\begin{lemma}
Let $\xi$ be the output of Algorithm~\ref{alg:ee} and $f^*$ be defined as above. 
Then
$$
\PPr{\Big|\xi-||f^*||_2^2\Big|\le14\eps||f||_2^2}\ge\frac{7}{8}.
$$
\end{lemma}
\begin{proof}
Let $g:\F_2^{\,n}\to\mathbb{R}$ be the $s$-sparse function defined by setting
$$
\hat{g}(\beta^*_{f|a+H})=\hat{f}(\beta^*_{f|a+H})
$$
for each $a\in S^*$ to be the only non-zero Fourier coefficients of $f^*$. 
Then by triangle inequality, 
$$
||f^*-h||_2\le||f^*-g||_2+||g-h||_2.
$$
Recall that $\cE_1 \ge  \dots \ge \cE_{2^n}$ are the true values of the energies of the $2^n$ Fourier coefficients corresponding to function $f : \F_2^{\,n} \to\mathbb{R}$ and $y^*_i$ is the contribution to the energy of the $i$\th bucket from the largest Fourier coefficient hashing into this bucket. 
Let $S$ be the set of indices corresponding to the buckets with nonzero energy in $f^*$ and $g$ and observe that $|S|\le s$. 
Thus, $||f^*-g||^2_2$ is at most $\sum_{i\in S}(y_i-\cE_i)$, where $y_i$ is the total energy in the $i$\th bucket.  
By Corollary~\ref{cor:hashing-error}, $\sum_{i\in S}(y_i-\cE_i)\le 5\eps^2||f||^2_2$ with probability at least $\frac{15}{16}$. 

On the other hand, $||g-h||_2^2\le16\eps^2||f||^2_2$ with probability at least $\frac{15}{16}$ by Lemma~\ref{lem:median:estimate} and Markov's inequality. 
Thus, $||f^*-h||_2\le\left(\sqrt{5}+4\right)\eps||f||_2\le7\eps||f||_2$ and by Lemma~\ref{lem:output:acc:first}, $\Big|\xi-||f^*||^2_2\Big|\le14\eps||f||_2^2$ with probability at least $\frac{7}{8}$.
\end{proof}

By Lemma~\ref{lemma:accuracy:queries}, it suffices to use $\O{\frac{s}{\eps^4}\|f\|_2^2}$ queries to bound the expected squared error of an estimator.  
Since Algorithm~\ref{alg:ee} take the median of $\ell=\O{\log\frac{1}{\eps}}$ estimators to bound the failure probability by a constant, then the total number of queries is $\O{\frac{s}{\eps^4}\|f\|_2^2\log\frac{1}{\eps}\log\frac{1}{\delta}}$ to obtain failure probability $1-\delta$. 
Hence, the query complexity follows as we assume $\|f\|_2^2 = 1$. 

Algorithm~\ref{alg:ee} runs through $\ell=\O{\log\frac{1}{\eps}}$ iterations, each time sampling $f$ at $\gamma=\O{\frac{s}{\eps^4}}$ pairs of points and updating each of the $2^d=\O{\frac{s}{\eps^4}}$ cosets. 
Hence, Algorithm~\ref{alg:ee} runs in $\O{\frac{s^2}{\eps^8}\log\frac{1}{\eps}}$ time. 
To boost the failure probability up to $1-\delta$, the total running time is $\O{\frac{s^2}{\eps^8}\log\frac{1}{\eps}\log\frac{1}{\delta}}$. 

We do not attempt to optimize runtime in Algorithm~\ref{alg:ee}, as further optimizations can be made using standard sparse Hadamard transform techniques, e.g. page 163 in~\cite{goldreich2000modern} or in \cite{levin1995randomness, ericpricecommunication} to update the empirical estimation of each coset, which improves the total running time to $\O{\frac{s}{\eps^4}\log\frac{s}{\eps^4}\log\frac{1}{\eps}\log\frac{1}{\delta}}$.
\section{Lower Bounds for $\ell_2^2$-Testing of $s$-Sparsity}
To the best of our knowledge the only lower bound known for the $s$-sparsity testing problem is due to~\cite{GOSSW11}.
Formally, they construct a hard distribution that is far from $s$-sparse in Hamming distance but since the support of the distribution is Boolean functions this also implies a lower bound under $\ell_2^2$. 
Under $\ell_2^2$-distance their Theorem 2 can be restated as follows:
\begin{theorem}[Lower bound for $\ell_2^2$ testing of Fourier sparsity~\cite{GOSSW11}]
Fix any constant $\tau > 0$. 
Let $C(\tau) = \O{\log 1/\tau}$ and $s \le 2^{n/C(\tau)}$. 
There exists a constant $c(\tau)$ so that any algorithm, which given non-adaptive query access to  $f \colon \mathbb F_2^n \to \{-1,1\}$, that distinguishes $s$-sparse functions from functions that are $c(\tau)$-far from $s$-sparse in $\ell_2^2$ distance with probability at least $2/3$ requires $\Omega(\sqrt{s})$ queries.
\end{theorem}

\noindent
Below we extend this result to larger values of $s$ for non-adaptive testers of real-valued functions.

\newcommand{\g}{\mathbf g}
\renewcommand{\S}{\mathbf S}

\subsection{$\Omega(\sqrt{s}$) Lower Bound for Non-adaptive Testers}\label{sec:lb}

We show a lower bound by designing two distributions $\mathcal D_{YES}$ and $\mathcal D_{NO}$, the former supported on the class of interest and the latter being far from it, such that the total variation distance between these distributions restricted to the query set is at most $\delta$.
This implies that the query set cannot distinguish the two distributions with probability greater than $\frac{1+\delta}{2}$. 

\begin{definition}[Total Variation Distance]
The total variation distance between two random variables $P_1$ and $P_2$ with corresponding probability density functions $p_1(x),p_2(x)\in\mathbb{R}^n$ is defined as $d_{TV}(P_1,P_2)=\frac{1}{2}\int_{\mathbb{R}^n}|p_1(x)-p_2(x)|\,dx$.
\end{definition}

\begin{theorem}
\label{thm:lb:na}
For any $s \le 2^{n - 1}$, there exists a constant $c > 0$ such that any non-adaptive algorithm given query access to $f \colon \mathbb F_2^n \to \mathbb R$ such that $\|f\|_2^2 = 1\pm \epsilon$ that distinguishes whether $f$ is $s$-sparse or $f$ is $\frac 13$-far from $s$-sparse in $\ell_2^2$ with probability at least $2/3$ has to make at least $c\sqrt{s}$ queries to $f$. 
\end{theorem}
\begin{proof}
We define two distributions $\mathcal D_{YES}$ and $\mathcal D_{NO}$ where $\mathcal D_{YES}$ is supported on $s$-sparse functions only and $\mathcal D_{NO}$ is supported on functions that are far from $s$-sparse. 
Then by Yao's principle it suffices to show that if the size of the query set $Q$ is at most $c \sqrt{s}$ then the total variation distance between the two distributions restricted on the query set $d_{TV}(\mathcal D_{YES} (Q), \mathcal D_{NO}(Q)) < 1/3$.
	
We now define the $\mathcal D_{YES}$ distribution.
For each $z \in 2^{[n]}$ let $\g_z \sim N(0,1)$ be an independent zero mean and unit variance Gaussian random variable.
Let $\S \subseteq 2^{[n]}$ be a random subset of fixed size $s$ chosen uniformly at random from the collection of all subsets of size exactly $s$.
Our distribution $\mathcal D_{YES}$ corresponds to a random family of functions $f_{\S}$ defined as follows:
$$f_{\mathbf S}(x) := \frac{1}{\sqrt{s}} \sum_{z \in \S} \g_z \chi_z(x).$$
The distribution $\mathcal D_{NO}$ is defined similarly, except that we fix $S = 2^{[n]}$, i.e. we set:
$$f(x) = \frac{1}{2^{n/2}} \sum_{z \in 2^{[n]}} \g'_z \chi_z(x),$$
where $\g'_z \sim N(0,1)$ are again independent and identically distributed standard normal variables. 

Note that by standard Chernoff bounds with high probability functions sampled from both distributions satisfy $\|f\|_2^2 = 1\pm\epsilon$. 
Furthermore by Chernoff bounds, with high probability functions in the support of $\mathcal D_{NO}$ are at least $\frac13$-far in $\ell_2^2$ from $s$-sparse for $s \le 2^{n - 1}$ (their expected distance is at least $1/2$). 
Consider any non-adaptive randomized algorithm that makes $q$ queries.
By Yao's principle we can fix the set of queries to form a set $Q \subseteq \mathbb F_2^n$ or size $q$.
The values of $f_{\S}$ on $Q$ form a vector with (possibly correlated) zero mean Gaussian entries.
	
Fix any $S$ of size $s$.
If $x = y$ then we have: 
$$\mathbb E_{\g}[f_{S}(x)f_{S}(y)]=\mathbb E_{\g}[f_{S}(x)^2] = \frac1s \mathbb E_{\g}\left[\left(\sum_{z_1 \in S} \g_{z_1} \chi_{z_1}(x)\right)^2\right] =  \frac1s \left(\sum_{z_1 \in S} \mathbb E_{\g}[\g_{z_1}^2]\right)= 1.$$
Computing the values of the off-diagonal entries in the covariance matrix of $f_{S}$ for $x \neq y$ we have:
\begin{align*}
\mathbb E_{\g}[f_{S}(x) f_{S}(y)] 
&= \frac1s \mathbb E_{\g}\left[\sum_{z_1 \in S} \g_{z_1} \chi_{z_1}(x) \sum_{z_2 \in S} \g_{z_2} \chi_{z_2}(y)\right] \\
&= \frac1s \mathbb E_{\g}\left[\sum_{z \in S} \g_z^2 \chi_z(x) \chi_z(y) + \sum_{z_1 \neq z_2 \in S} \g_{z_1} \chi_{z_1}(x) \g_{z_2} \chi_{z_2}(y)\right] \\
&= \frac1s \left(\sum_{z \in S} \chi_z(x) \chi_z(y) \mathbb E_{\g}\left[ \g_z^2\right]  +  \sum_{z_1 \neq z_2 \in S}\chi_{z_1}(x) \chi_{z_2}(y) \mathbb E_{\g}\left[ \g_{z_1}  \g_{z_2} \right] \right)\\
&= \frac1s \left(\sum_{z \in S} \chi_z(x) \chi_z(y)   +  \sum_{z_1 \neq z_2 \in S}\chi_{z_1}(x) \chi_{z_2}(y) \mathbb E_{\g}\left[\g_{z_1}] \mathbb E_{\g}[\g_{z_2} \right] \right)\\
&= \frac1s \sum_{z \in S} \chi_z(x) \chi_z(y)
\end{align*}
	
Let $\xi_1, \dots, \xi_q$ be the inputs in the query set $Q$.
For any fixed $z \in 2^{[n]}$ define $a_z \in \{-1,1\}^q$ to be a column vector with entries $a_{z,i} = \chi_z(\xi_i)$.
Then the covariance matrix of $f_\S(\xi_1), \dots f_\S(\xi_q)$ under the distribution $\mathcal D_{YES}$  is given by a random family of matrices $M_{\S} \in \mathbb R^{q \times q}$ defined as follows: 
$$M_{\S} = \frac{1}{s} \sum_{z\in \S} a_z a_z^T.$$
	
Similarly for $\mathcal D_{NO}$ the covariance matrix of $f(\xi_1), \dots, f(\xi_q)$ is $\frac{1}{2^n} \sum_{z \in 2^{[n]}} a_z a_z^T = I$.
	
The following standard fact allows to bound the total variation distance between two zero mean Gaussians with known covariance matrices.
\begin{fact}(See e.g. Corollary 2.14 in~\cite{DKKLMS16})
Let $\delta > 0$ be sufficiently small and let $\mathcal N(0, \Sigma_1)$ and $\mathcal N(0, \Sigma_2)$ be normal distributions with zero mean and covariance matrices $\Sigma_1$ and $\Sigma_2$ respectively.  If $\|I - \Sigma_2^{-1/2} \Sigma_1 \Sigma_2^{-1/2}\|_F \le \delta$ then:
$$d_{TV}(\mathcal N(0, \Sigma_1), \mathcal N(0, \Sigma_2)) \le \O{\delta}.$$
\end{fact}
	
Using the above fact and setting $\Sigma_1 = M_{\S}$ and $\Sigma_2 = I$ in order to show an upper bound on the total variation distance it suffices to bound the expected Frobenius norm of the difference $\mathbb E_{\S}\left[\|I - M_{\S}\|_F\right]$.
	
We have:
\begin{align*}
\mathbb E_{\S}\left[\left\|I - \frac{1}{s} \sum_{z\in \S} a_z a_z^T\right\|_F\right] &= \mathbb E_{\S} \left[\sum_{1 \le i,j \le q} \left(\delta_{ij} - \frac1s \sum_{z \in \S} \chi_z(\xi_i) \chi_z(\xi_j)\right)^2\right] \\
&=  \sum_{1 \le i \le q} \mathbb E_{\S} \left[\left(1 - \frac1s \sum_{z \in \S} \chi_z(\xi_i)^2\right)^2\right] + \sum_{1 \le i \neq j \le q} \mathbb E_{\S}\left[\left( \frac1s \sum_{z \in \S} \chi_z(\xi_i) \chi_z(\xi_j)\right)^2\right] \\
&= \frac{1}{s^2} \sum_{1 \le i \neq j \le q} \mathbb E_{\S}\left[\left( \sum_{z_1 \in \S} \chi_{z_1}(\xi_i) \chi_{z_1}(\xi_j)\right)\left( \sum_{z_2 \in \S} \chi_{z_2}(\xi_i) \chi_{z_2}(\xi_j)\right)\right] \\
&= \frac{1}{s^2} \sum_{1 \le i \neq j \le q} \mathbb E_{\S}\left[ \sum_{z \in \S} \chi_{z}(\xi_i)^2 \chi_{z}(\xi_j)^2\right] + \mathbb E_{\S}\left[ \sum_{z_1 \neq z_2 \in \S} \chi_{z_1}(\xi_i) \chi_{z_1}(\xi_j)\chi_{z_2}(\xi_i) \chi_{z_2}(\xi_j)\right] \\
& \le \frac{q^2}{s}
\end{align*}
	
Thus if $q < \sqrt{\delta s}$ we have $d_{TV}(\mathcal D_{YES}(Q), \mathcal D_{NO}(Q)) \le \O{\delta}$. By picking $\delta$ to be a sufficiently small constant it follows that no algorithm that makes less than $c \sqrt{s}$ queries for some constant $c > 0$ can distinguish $\mathcal D_{YES}$ and $\mathcal D_{NO}$ with high probability.
\end{proof}

\section*{Acknowledgements}

The authors would like to thank Piotr Indyk and Eric Price for multiple helpful discussions of this topic as well as Andrew Arnold, Arturs Backurs, Eric Blais, Michael Kapralov and Krzysztof Onak for their participation in earlier versions of this work.

\def\shortbib{0}
\bibliographystyle{alpha}
\bibliography{references}

\appendix
\section{Appendix}
\subsection{Basic Facts}
\label{app:facts}
\begin{fact}[Reduction of Property Testing to Energy Estimation of Top $s$ Fourier Coefficients]\label{fact:approx-to-testing}
Suppose we are given query access to some function $f:\F_2^{\,n}\to\mathbb{R}$ with $||f||_2^2=1$. 
Given an energy estimator of the top $s$ Fourier coefficients that uses $q_s(\eps)$ queries, there exists a property tester for $s$-sparsity with parameter $\eps$ that uses $q_s\left(\frac{\eps}{2}\right)$ queries, where $q_s(\cdot)$ is some function that depends on $\eps$.
\end{fact}
\begin{proof}
Let $\ksparse$ be the class of $s$-sparse functions mapping from $\F_2^{\,n}$ to $\mathbb{R}$. 
Trivially if $f\in\ksparse$, then the sum of the top $s$ Fourier coefficients is $||f||_2^2$ and so an $\frac{\eps}{2}$-energy estimator of the top $s$ Fourier coefficients outputs a value $\xi$ with $\xi\ge||f||_2^2-\frac{\eps}{2}||f||_2^2$.

On the other hand, if for any $s$-sparse function $g$, it holds that $||f-g||_2^2\ge\eps||f||_2^2$, then the energy of the top $s$ Fourier coefficients of $f$ is at most $(1-\eps)||f||_2^2$. 
Then an $\frac{\eps}{2}$-energy estimator of the top $s$ Fourier coefficients outputs a value $\xi$ with 
\[\left|\xi-\max_{|S|=s}\sum_{\alpha\in S}\hat{f}(\alpha)^2\right|\le\frac{\eps}{2}||f||_2^2,\]
so the energy estimator outputs a value $\xi$ with $\xi\le||f||_2^2-\frac{\eps}{2}||f||_2^2$. 

Thus, the energy estimator can differentiate whether $f\in\ksparse$ or $f$ is $\eps$-far from $s$-sparsity, using $q_s\left(\frac{\eps}{2}\right)$ queries.
\end{proof}

\subsection{Poisson Summation Formula}
\label{app:poisson}
Recall the proof of the Poisson summation formula:
\begin{proofof}{Proposition \ref{prop:projection-formulas}}
For any $z \in \F_2^{\,n}$, we have that
\begin{align*}
f|_{a+H}(z) &= \underset{x \in H^\perp}{\mathbb{E}}\Big[ \sum_{\beta \in \F_2^{\,n}} \hat{f}(\beta) \chi_{\beta}(x+z) \cdot \chi_a(x) \Big] \\
&= \sum_{\beta \in \F_2^{\,n}} \hat{f}(\beta) \chi_\beta(z) \cdot \underset{x \in H^\perp}{\mathbb{E}}\left[ \chi_{\beta + a}(x)\right].
\end{align*}
Since $\underset{x \in H^\perp}{\mathbb{E}}\Big[\chi_{\beta + a}(x)\Big]$ equals 1 when $\beta+a \in H$ and 0 otherwise, we obtain 
$$
f|_{a+H}(z) = \sum_{\beta \in a+H} \hat{f}(\beta) \chi_\beta(z) 
$$
and hence
\begin{align*}
\widehat{f|}_{a + H}(\alpha) &= \underset{x \in \mathbb F^n_2}{\mathbb{E}}\left[f|_{a + H}(x) \chi_\alpha(x)\right] = \underset{x \in \mathbb F^n_2}{\mathbb{E}}\left[\sum_{\beta \in a + H} \hat f(\beta) \chi_\beta(x) \chi_\alpha(x)\right] \\
&= \sum_{\beta \in a + H} \left(\hat f(\beta) \underset{x \in \mathbb F^n_2}{\mathbb{E}}\left[\chi_\beta(x) \chi_\alpha(x)\right]\right) = \hat f(\alpha).
\end{align*}
\end{proofof}

\end{document}